\title{Test Martingales for bounded random variables}
\author{Harrie Hendriks}
\date{}

\documentclass[12pt]{article}
\newif\ifappendix\appendixfalse
\usepackage{graphicx}
\usepackage{amsmath}
\usepackage{amsfonts}
\usepackage{amssymb}
\usepackage[english]{babel}
\usepackage[gen]{eurosym}
\newtheorem{theorem}{Theorem}

\newtheorem{corollary}[theorem]{Corollary}

\newtheorem{lemma}[theorem]{Lemma}

\newtheorem{remark}{Remark}

\newenvironment{proof}[1][Proof]{\noindent\textbf{#1}\quad}{\ \rule{0.5em}{0.5em}}
\def\text{\mbox}
\def\R{{\mathbb R}}

\def\P{{\mathbb P}}

\def\E{{\bf E}}

\def\Var{{\rm Var}}

\def\Alt{\hbox{\rm Alt}}
\begin{document}

\maketitle

\begin{abstract}
\iffalse
Test martingales have been proposed as a more intuitive approach to hypothesis testing
than the classical approach based on the \emph{p-value} of a sample.
%Test martingales lead to sequential test procedures that allow 
If one is interested in a test procedure that in some sense allows 
you to do minimal work
when lucky and to do further work if less lucky
the test martingale technique allows to accomodate the new sample results,
much like Bayesian decision techniques do.
Unlike Bayesian techniques the test martingale technique allows for a classical frequentist interpretation
of the test result.
\\
\fi
Given a random sample from a random variable $T$ which is bounded from above, $T\le\tau$ a.s., 
we define processes that are positive supermartingales if $\E(T)\ge\mu$.
Such processes are called test martingales.
%under a hypothesis of the type $H_0:\E(T)\ge\mu$. 
Tests of the supermartingale hypothesis implicitly test the hypothesis $H_0:\E(T)\ge\mu$.
We construct test martingales that lead to tests with power 1. 
We also construct confidence upper bounds.
We extend the techniques to testing $H_0:\E(T)=\mu$ and constructing confidence intervals.
\\
In financial auditing random sampling is proposed as one of the possible techniques to gather
enough assurance to be able to state that there is no 'material' misstatement in a financial report.
The goal of our work is to provide a mathematical context 
that could represent such process
of gathering assurance by means of repeated random sampling.
\end{abstract}

% Key words and phrases: 


%%%%%%%%%%%%%%%%%%%%%%%%%%%%%%%%%%%%%%%%%%%%%%%%%%%%%%%%%%%%%%%%%%%%%%%%%%%%%%
\section{Introduction}\label{Introduction}
We are inspired by Gr\"unwald \cite{PG} and Shafer et al \cite{SSVV}
who recall the relationship between sequential probability ratio tests (\cite{W}) and martingale theory.
In \cite{PG} the test martingale concept is explicitly announced
as a contribution to the current discussion about the interpretation of \emph{p-value} in scientific literature.
The cited works mainly describe tests concerning the parameters in a parametrized family of 
probability distributions. 
We will describe tests concerning the expectation value of a random variable,
under the only assumption that the support (of the probability distribution) of the random variable 
is bounded from above and/or below 
%is contained  in some given interval 
depending on the null hypothesis. %$[\tau_0,\tau_1]$.
\\
\\
In this paper we hope to reach not only statisticians with a reasonable background in probability,
but also applied statisticians. 
That is why we will explain some notions from probability.
We will say that some event is almost sure, or a.s., 
if its probability is 1 with respect to the relevant
probability distribution(s). 
The term random variable may be abbreviated to rv. 
A random variable $Z$ will be said \emph{integrable} if its expected value $\E(Z)$ exists (and is finite)
and will be called \emph{positive} if $Z\ge0$ a.s..
A sequence of rv's $\{T_k\}_{k=1}^\infty$, 
%i.e. 
$T_1,T_2,T_3,\ldots$,
is a random sample or an iid (independent identically distributed) sample of $T$ if it is an collection of independent rv's 
and each $T_k$ has the same distribution as $T$.
\\
\\
In the context of financial auditing
we have in mind that $T$ is defined on some population $\Omega$, 
say a finite set $\Omega=\{\omega_1,\ldots,\omega_L\}$, in the sense that given $\omega\in \Omega$ 
there is a well defined procedure to determine its value $T(\omega)\in\R$.
The auditor has to assure himself that $\Omega$ is well defined and that the procedure to determine
a $T$-value is practically feasible.
One is interested in a characteristic of $T$ that can be interpreted as 
the expected value $\E(T)=\sum_{\omega\in\Omega}T(\omega)p(\omega)$
with respect to a probability density $p$ on $\Omega$
(i.e. 
 $\forall\omega:p(\omega)\ge0$ and $\sum_{\omega\in\Omega} p(\omega)=1$).
\\
For example $\omega_1,\ldots,\omega_L$ are identifiers of items underlying a financial report.
For $\omega\in\Omega$ one has its book value $B(\omega)>0$, the audited value $A(\omega)$
and the so-called tainting $T(\omega)=(B(\omega)-A(\omega))/B(\omega)$.
In this context one usually knows that $0\le A(\omega)\le B(\omega)$, so that $1\ge T(\omega)\ge0$.
The total book value is $B_{\hbox{tot}}=\sum_{\omega\in\Omega}B(\omega)$. 
The total misstatement equals
$$\sum_{\omega\in\Omega}(B(\omega)-A(\omega))=B_{\hbox{tot}}\sum_{\omega\in\Omega}T(\omega)B(\omega)/B_{\hbox{tot}}=B_{\hbox{tot}}\sum_{\omega\in\Omega}T(\omega)p(\omega),$$
where $p(\omega)=B(\omega)/B_{\hbox{tot}}$ satisfies the properties of a probability density.
\\
The problem is that $\Omega$ is a large set, 
so that it is not practical to determine all $T$-values.
Given a number $z\in[0,1]$ one may associate to it that item $w(z)=\omega_\ell\in \Omega$ such that
$\sum_{i=1}^{\ell-1}p(\omega_i)<z\le\sum_{i=1}^{\ell}p(\omega_i)$.
A random sample $T_1,T_2,\ldots$ of $T$ can be constructed, by using a random number generator 
that yields a random sample of
numbers $z_1,z_2,\ldots$, uniformly distributed in [0,1], and 
associate to it $T_1=T(w(z_1)), T_2=T(w(z_2)),\ldots$.
Notice that all $\omega\in\Omega$ will occur (almost surely) infinitely often
in the sequence $w(z_1),w(z_2),\ldots$.
\\
\\
A  process (in discrete time) $\{X_k\}_{k=0}^\infty$ is a sequence of 
rv's $X_0,X_1,X_2,\ldots$.
The index $k$ is referred to as time, and we will speak about time $k$.
The notion of a filtration $\{\mathcal F_k\}_{k=0}^\infty$ is used to formalize the notion of time.
Here $\mathcal F_k$ is a $\sigma$-algebra which represents all the information that is available 
at time $k$.
The process $\{X_k\}_{k=0}^\infty$ is adapted to $\{\mathcal F_k\}_{k=0}^\infty$
if the variables $X_0,X_1,\ldots,X_k$ are $\mathcal F_k$-measurable, that is, their values 
can be observed, measured, %and measurable
by time $k$.
The process is \emph{integrable} (resp. \emph{positive}) if each rv $X_k$ is integrable (resp. positive).
For $\ell\ge k$ the conditional expectation $\E(X_\ell\mid\mathcal F_k)$
is an rv that is $\mathcal F_k$-measurable. % at time $k$.
If the process $\{X_k\}_{k=0}^\infty$ is adapted to $\{\mathcal F_k\}_{k=0}^\infty$, 
$\E(X_\ell\mid\mathcal F_k)$ %it 
represents the expected value of $X_\ell$, given the information at time $k$
which includes the observed values of $X_0,\ldots,X_k$.
It holds that
$\E(X_\ell\mid\mathcal F_k)=X_\ell$ for $\ell\le k$.
On the other hand
the conditional probability of rv $X_\ell$ with respect to the trivial $\sigma$-algebra (no information)
corresponds to the ordinary notion of expected value.
It holds that $\E(X_\ell)=\E(\E(X_\ell\mid\mathcal F_k))$.
\\
Given a null hypothesis $H_0$ and a hypothesis test for $H_0$, 
its \emph{size} at some rv $T$ that satisfies $H_0$ is the probability of an error of Type I,
that $H_0$ is rejected by the test,
when applied to $T$.
Without reference to an rv, 
the size of a test equals the maximum size possible at any rv satisfying the null hypothesis.
\\
Given $0\le\nu\le1$, by $\Alt(\nu)$  we
denote the alternative distribution with values 0 and 1 and expected value $\nu$, 
in particular the probability of 1 (resp. 0) is $\nu$ (resp. $(1-\nu)$).
An argument that is used a few times is the following variation on Jensen's inequality:
\begin{lemma}\label{Jensen+}
Let $t\mapsto f(t)$, $\tau_0\le t\le\tau_1$, be a concave function and $T$ a random variable such that
$\tau_0\le T\le\tau_1$ a.s..
Let $T^A$ be the alternative distribution with values $\tau_0$ and $\tau_1$ such that $\E(T^A)=\E(T)$.
Then $\E(f(T^A))\le\E(f(T))\le f(\E(T))$.
\end{lemma}
\begin{proof}
Let $\overline f$ be the linear interpolation of $f$ at the points $\tau_0,\tau_1$,
$\overline f(t)=((\tau_1-t)f(\tau_1)+(t-\tau_0)f(\tau_0))/(\tau_1-\tau_0)$.
Then $\overline f(t)\le f(t)$ for $t\in[\tau_0,\tau_1]$
and $\E(f(T^A))=\E(\overline f(T^A))=\E(\overline f(T))\le\E(f(T))$.
Jensen's inequality yields $-f(\E(T))\le \E(-f(T))=-\E(f(T))$.
\end{proof}
\\
\\
In section \ref{Test martingales} we expose the so-called maximal lemma 
and show how it leads to a
test that a random process is a supermartingale.
Suppose given an upper bound $\tau$
and a null hypothesis $H_0:\E(T)\ge\mu$ about rv's $T$ such that $T\le\tau$ a.s..
In section \ref{H0} we develop a method to construct a process $\{M_k\}_{k=0}^\infty$, given a
random sample $\{T_k\}_{k=1}^\infty$ of $T$ with $T\le\tau$ a.s.,
which will be a positive supermartingale, if $T$ satisfies $H_0$.
Such a process is called a test martingale.
In section \ref{Behavior} we study the behavior of test martingales, depending on $T$.
In section \ref{conf reg} we apply the technique to the construction of confidence upper bounds
and confidence intervals.

\section{Test martingales}\label{Test martingales}

Suppose given an integrable process $\{X_k\}_{k=0}^\infty$, 
adapted to the filtration $\{\mathcal F_k\}_k$.
The $\sigma$-algebra $\mathcal F_k$ represents the information available at time $k$, including
the values of $X_i$ for $i\le k$.
Then define
\begin{align*}
X_k^*&=\sup_{i\le k}X_i
\\
X_\infty^*&=\sup_kX_k
\end{align*}
The process $\{X_k\}_k$ is a \emph{supermartingale} 
if $X_k\ge\E[X_\ell\mid \mathcal F_k]$ for $0\le k\le \ell<\infty$.
(It is called a martingale if $X_k=\E[X_\ell\mid \mathcal F_k]$ for $0\le k\le \ell<\infty$
and a submartingale if $X_k\le\E[X_\ell\mid \mathcal F_k]$ for $0\le k\le \ell<\infty$.)
Suppose $\{X_k\}_k$ is a positive supermartingale.
The 'maximal lemma' (\cite[Ch. V.3, Corollary 21]{DM} and 
\cite[Ch. 5 Exercise 5.7.1]{D})
% Third Ed. Ch. 4, Exercise 7.1
implies that
$$
\forall\lambda\ge0:\lambda\,\P\{X_\infty^*\ge\lambda\}\le\E[X_0].
$$
We apply this result in the following somewhat weaker form: % of which we give an indication of proof.

\begin{lemma}\label{maximal lemma}
If $\{X_k\}_k$ is a positive supermartingale,
 then
$$\forall\lambda\ge0:\lambda\,\P\{\exists \ell: X_\ell\ge\lambda\}\le\E[X_0].$$
\end{lemma}

\begin{proof}
%A proof %of the maximal lemma 
%could run along the following lines: 
Consider the random variable $N$ which is 
the first time $k$ that the process $\{X_k\}_k$ reaches or exceeds level $\lambda$,
or $\infty$ if the process does not exceed level $\lambda$.
$N$ is a stopping time.
If one stops the supermartingale $\{X_k\}_k$ at that time, 
the stopped process $\{X_{k\wedge N}\}_k$ is still a supermartingale.
Thus for any $\ell$  we have $\E[X_0]=\E[X_{0\wedge N}]\ge\E[X_{\ell\wedge N}]$ while
for the positive random variable $X_{\ell\wedge N}$ we have
$\E[X_{\ell\wedge N}]\ge\lambda\P\{X_{\ell\wedge N}\ge\lambda\}$. %\ge\lambda\P(X_{t\wedge S}\ge\lambda)$ 
%(Cf. \cite[Ch. 5 Exercise 5.7.1]{D}). 
The lemma then follows since the events 
$\{X_{\ell\wedge N}\ge\lambda\}=\{N\le\ell\}$ form an increasing sequence for increasing $\ell$
whose union is $\{N<\infty\}=\{\exists \ell:X_\ell\ge\lambda\}$.
\end{proof}

\begin{remark}\label{sharpness}
Notice that Lemma \ref{maximal lemma} yields a sharp bound for $\P\{\exists k:X_k\ge\lambda\}$ 
in the case of the martingale associated with
a random sample $Z_1,Z_2,\ldots$ of an alternative distribution $Z\sim \Alt(\mu)$ with $0<\mu<1$,
where one takes 
$$X_0=1 \hbox{ and }X_k=X_{k-1}\cdot(1-(Z_k-\mu)/(1-\mu)).$$
In the event $\{Z_1=Z_2=\ldots=Z_k=0,Z_{k+1}=1\}$ we have $X_i=(1-\mu)^{-i}$ for $i\le k$
and $X_\ell=0$ for $\ell\ge k+1$.
Let $n$ be such that $(1-\mu)^{1-n}<\lambda\le(1-\mu)^{-n}$.
Then $\P\{\exists k:X_k\ge\lambda\}=\P\{Z_1=\ldots=Z_n=0\}=(1-\mu)^n\le\lambda^{-1}$
with equality if $\lambda=(1-\mu)^{-n}$.
\end{remark}
We follow \cite{PG} and \cite{SSVV} where 
the significance of the above ideas 
for
statistical hypothesis testing is worked out. 
%\\
Be given a statistical hypothesis $H_0$.
% Assuming the truth of a statistical hypothesis $H_0$,
A \emph{test (super)martingale} (for $H_0$) is a process $\{X_k\}_{k=0}^\infty$
% satisfying $\forall t:X_t\ge0$ a.s.,
such that, if $H_0$ is satisfied, 
$\{X_k\}_k$ is a positive supermartingale and $\E[X_0]\le1$.
Be given a significance level $\alpha$, $0<\alpha\le1$.
A test consists of observing $X_1,X_2,\ldots,X_n$, where $n$ is allowed to depend on
anything including the process itself, 
e.g. that $X_n^*\ge\lambda$ if that ever happens.
We reject $H_0$
if $X_n^*\ge1/\alpha$, and otherwise we do not reject $H_0$.
The size of such test satisfies
$$
\P\{X_n^*\ge1/\alpha\}\le\P\{\exists k: X_k\ge1/\alpha\}\le\alpha\E(X_0)\le\alpha.
$$
%%%%%%%%%%%%%%%%%%%%%%%%%%%%%%%%%%%%%%%%%%%%%%%%%%%%%%%%%%%%%%%%%%%%%%%%%%%%%%
\section{Null hypothesis and supermartingales}\label{H0}

We will construct test martingales to test the null hypothesis $\E(T)\ge\mu$.
Given $\tau>\mu$, let us find functions $f:\R\to\R$ such that
for all integrable random variables $T$ the following holds
$$[T\le \tau\hbox{ a.s. \& }\E(T)\ge\mu]\Rightarrow [f(T)\ge0\hbox{ a.s. \& }\E(f(T))\le 1].$$
Examples are $f(t)=a-b\,t$ with $b\ge0$ such that $a-b\,\tau\ge0$
and $a-b\mu\le1$.
In particular one needs $b\,\tau\le a\le 1+b\mu$ and $b(\tau-\mu)\le1$.
So we have $0\le b\le1/(\tau-\mu)$ and for optimality reasons we will prefer $a=1+b\mu$.
Thus we find the following functions for $0\le c\le 1$
\begin{equation*}%\label{factor}
f(t)=1-c\cdot\frac{t-\mu}{\tau-\mu}.
\end{equation*}

\begin{remark}
The condition on the essential supremum 
of $T$, $T\le\tau$ a.s., for some $\tau\in\R$, is necessary to make it possible that
$\E(f(T))>1$ for some random variable $T$ with $\E(T)<\mu$.
\end{remark}
One can see this as follows.
If $f$
is a function such that $f(x)\ge0$ for all $x$ and such that for all rv's $X$ with $\E(X)\ge\mu$
it holds that $\E(f(X))\le1$, then $f(x)\le1$ for all $x$.
Namely, suppose $f(x_1)>1$. Then $x_1<\mu$. Let $x_2>\mu$ be such that the linear interpolation
of $f$ between $x_1$ and $x_2$ has value larger than 1 at $\mu$ 
(e.g. take $x_2>\mu+(\mu-x_1)/(f(x_1)-1)$). 
Then the rv $X$ with expected value $\mu$ and support $\{x_1,x_2\}$ 
has $\E(f(X))>1$. $\blacksquare$
\\
\\
We present a variation on the condition $T\le\xi$ a.s.

\begin{remark}
Let $g:\R\to\R$ be an increasing strictly convex function. 
Replace the condition $T\le\tau$ a.s. with the condition
$\E(g(T))\le\tau$ for some given $\tau\ge g(\mu)$.
Then consider the strictly convex function $f(x)=a-b(x-\mu)-c(\tau-g(x))$, with $a,b,c>0$,
such that $f(x)\ge0$ for all $x$.
If $f$ reaches its minimum at $x_0$ then it would be the solution of $g'(x_0)=b/c$, and we 
would require that
$a-b(x_0-\mu)-c(\tau-g(x_0))\ge0$.
Moreover under the additional condition $\E(T)\ge\mu$ we require that $\E(f(T))\le1$, that is
$a\le1$.
For testing purposes it is optimal to choose $a=1$.
\end{remark}
%
\iffalse
mu<-0.05
eta<-exp(1)
a<-1
x0<--5
uniroot(function(cc)a-cc*exp(x0)*(x0-mu)-cc*(eta-exp(x0)),c(0,10))
cc0<-a/(exp(x0)*(x0-mu)+eta-exp(x0))
x<-seq(-500,3,length=1000);plot(x,a-cc0*exp(x0)*(x-mu)-cc0*(eta-exp(x)),type='l')
abline(h=c(0,1),v=mu,col='red')
x0<--1
cc0<-a/(exp(x0)*(x0-mu)+eta-exp(x0))
x<-seq(-50,3,length=1000);plot(x,a-cc0*exp(x0)*(x-mu)-cc0*(eta-exp(x)),type='l')
abline(h=c(0,1),v=mu,col='red')
x0<-+5
cc0<-a/(exp(x0)*(x0-mu)+eta-exp(x0))
x<-seq(-1,1,length=1000);plot(x,a-cc0*exp(x0)*(x-mu)-cc0*(eta-exp(x)),type='l')
abline(h=c(0,1),v=mu,col='red')
\fi
Let $\tau\in\R$ be given, and suppose $T$ is an integrable random variable 
such that $T\le\tau$ a.s.. 
For $\mu<\tau$ 
consider the null hypothesis
$$H_0:\E(T)\ge\mu.$$
We would like to construct a 
test martingale for $H_0$.
\\
\\
Consider a random sample $T_1,T_2,\ldots$ of the random variable $T$.
It defines a filtration $\{\mathcal F_k^o\}_k$ by $\sigma$-algebras $\mathcal F_k^o=\sigma(T_1,\ldots,T_k)$ for $k\ge1$ and the trivial $\sigma$-algebra $\mathcal F_0$.
Under the null hypothesis we get a supermartingale as follows.
We let $M_0=1$. 
At time $(k-1)$ 
the variables $T_1,\ldots,T_{k-1}$ and $M_{k-1}$ are observed, having values $t_1,\ldots,t_{k-1}$ and $m_{k-1}$,
and one defines 
\begin{equation}\label{test martingale}
M_k=m_{k-1}\cdot\left(1-c_{k-1}\frac{T_k-\mu}{\tau-\mu}\right).
\end{equation}
Here $c_{k-1}$ may not depend on $T_k,T_{k+1},\ldots$ in any conceivable way
(more precisely, $c_{k-1}$ is 'predictable' at time $k$, i.e. $c_{k-1}$ is $\mathcal F_{k-1}$-measurable).
%In particular they have to be determined at the moment that the to be observed value of $T_k$ is not yet known. 
For our choice of filtration the functional dependence of $c_{k-1}$ on the observations of $T_1,\ldots,T_{k-1}$ 
should have been
fixed before any observation was available.
But see Remark \ref{filtration} for a broader, more practical class of filtrations.
The process $\{M_k\}_k$ is a test martingale for $H_0$.
We obtain a test with significance level $\alpha$ if we reject $H_0$ at a time $k$ with 
$M_k^*\ge1/\alpha$.
%where $b_{k-1}\mu$ is a reasonable estimate of $\E(T)$.

\begin{remark}\label{sharpness2}
(Cf. Remark \ref{sharpness})
If $\tau>\mu>0$, and one believes that $T=0$ a.s., one would choose $c_{k-1}=1$
in order to have $M_k$ as large as possible. 
We reject $H_0$ after $k$ observations when we observe $m_k=(\tau/(\tau-\mu))^k\ge1/\alpha$,
that is, $(1-\mu/\tau)^k\le\alpha$.
Unfortunately, if $\E(T)>0$ and one observes some $t_i=\tau$, then $M_k=0$ for $k\ge i$, 
and there is no hope to reject $H_0$ afterwards.
\end{remark}
If one considers a classical test for alternative distributions $\Alt(\nu)$ to test $H_0:\nu\ge\mu$
the smallest sample size needed for significance level $\alpha$ 
is the minimal $n$ such that $(1-\mu)^n\le\alpha$.
That is exactly the size where one may stop the test martingale, 
constructed with $\tau=1$ and $c_{k-1}=1$ for all $k$.
\\
As in \cite{PG} we will express the above construction in a \emph{gambling metaphor}, 
that we present as a
'martingale transform' of a supermartingale by a positive predictable process (see e.g. \cite[Thm. 5.2.5]{D}, \cite[Section 10.6]{Wpwm}).
We consider the hypothesis $H_0:\E(T)\ge\mu$ that we would like to reject.
Based on the random sample $T_1,T_2,\ldots$ of $T$, 
consider the process $X=\{X_k\}_k$ with $X_0=0$ and $X_k=X_{k-1}-(T_k-\mu)/(\tau-\mu)$.
Under $H_0$ the process $X$ is a supermartingale 
with respect to the filtration $\{\mathcal F_k=\sigma(T_1,\ldots,T_k)\}_k$.
Consider a lottery that takes place at time $k$ and pays out $(\tau-T_k)/(\tau-\mu)$
per unit stake, so that the net gain per unit stake is 
$-1+(\tau-T_k)/(\tau-\mu)=-(T_k-\mu))/(\tau-\mu)=X_k-X_{k-1}$. 
We start with an initial unit amount of capital $m_0=1$.
At time $k-1$ we have accumulated a capital of $m_{k-1}$ and
we decide to stake an amount of $H_k=c_{k-1}m_{k-1}$ in this lottery, $0\le H_k\le m_{k-1}$.
Then at time $k$ our capital will be 
$m_{k-1}+H_k(X_k-X_{k-1})=m_{k-1}(1-c_{k-1}(T_k-\mu)/(\tau-\mu)=m_k$.
%The gain process $\{\sum_{i=1}^kH_k(X_k-X_{k-1})\}_k$ 
%is called the martingale transform of the process $X$ by the 
%predictable process $H=\{H_k\}_k$.
If $T$ satisfies hypothesis $H_0$, we have a fair or loss-making game.
In particular, if we succeed in ending up with a large gain, we have reason to state that $\E(T)<\mu$.
With this metaphor it should be intuitively correct, that 
it is wrong to change the stake amount $c_{k-1}m_{k-1}$, after the observation of $T_k$.
Moreover, if one continues with a new gambling game to reject $H_0$, 
one has to start off with the capital accumulated from the original gambling game.
\\
\\
\emph{Comparison with Likelihood Ratio test.}
Suppose $T$ is an rv which is distributed Alt($\mu$) or Alt($\nu$) and 
consider the null hypothesis $H_0:T\sim\Alt(\mu)$ that $T$ is distributed Alt($\mu$).
The likelihood ratio test yields maximal power at a given probability of rejection
of the hypothesis in case $T\sim\Alt(\mu)$.
This test is using a test statistic of the form
$$
\Lambda_k=\frac{(1-\mu)^{k-r_k}\mu^{r_k}}{(1-\nu)^{k-r_k}\nu^{r_k}}, \hbox{ where }r_k=\sum_{i=1}^kT_i
\hbox{ is the number of successes}$$
and $H_0:\E(T)=\mu$ is rejected if $\Lambda_k$ is small enough.
In \cite{W} such a starting point is worked out to the so-called Sequential Probability Ratio Test.
If we would accomodate this in our test martingale setup, we would be tempted to use 
%$a_{k-1}=(1-\nu)/(1-\mu)$ and $b_{k-1}=(\mu-\nu)/(\mu(1-\mu))$,
$c_{k-1}=(\mu-\nu)/\mu$
so that one obtains $M_k=1/\Lambda_k$, which is a martingale under $H_0$, 
and we would reject $H_0$ if for some $k$ we have $M_k\ge1/\alpha$.
Notice the identity
$$1-\frac{\mu-\nu}{\mu}\cdot\frac{T_k-\mu}{1-\mu}=
%1-\frac{\mu-\nu}{\mu(1-\mu)}(T-\mu)=
\frac{1-\nu}{1-\mu}(1-T_k)+\frac{\nu}{\mu}T_k,
$$
indicating that this choice of $c_{k-1}$ corresponds to a linear interpolation 
between the (inverse) likelihood ratios for $T_k=0$ and $T_k=1$ in the above context of alternative distributions.
\\
\\
A reasonable choice for $c_{k-1}$ in the construction (\ref{test martingale}) of a test martingale
is that value of $c$ that maximizes
$m_{k-1}(c)=\prod_{i=1}^{k-1}(1-c(t_i-\mu)/(\tau-\mu))$, where some prudence is necessary to avoid $c_{k-1}=1$.
A strongly recommended possibility is to start with some probability density $\pi$ on $[0,1]$, 
typically the uniform probability density on [0,1].
Define $\{M_{k}(c)\}_k$ to be the test martingale based on the choice $c_{k-1}=c$, all $k$,
and consider the \emph{integrated test martingale} with respect to $\pi$:
\begin{align*}
M_k(\pi)&=\int_0^1M_k(c)\pi(c)dc=\int_0^1M_{k-1}(c)\cdot(1-c\,(T_k-\mu)/(\tau_1-\mu))\pi(c)dc
\\&=M_{k-1}(\pi)\cdot(1-c_{k-1}(T_k-\mu)/(\tau_1-\mu)), \hbox{ where }
\\
c_{k-1}&=\int_0^1 c\cdot M_{k-1}(c)\pi(c)dc/M_{k-1}(\pi).%=
%\int_0^1 c\cdot M_{k-1}(c)d\pi(c)/\int_0^1  M_{k-1}(c)d\pi(c.
\end{align*}
Notice that $c_{k-1}$ is the expectation 
of the probability density $f_{k-1}$ defined by $f_{k-1}(c)=M_{k-1}(c)\pi(c)/M_{k-1}(\pi)$, 
and that $M_{k-1}(c)$ is a log-concave function in $c$ (see Lemma \ref{L1}).
In case $\pi$ is the uniform probability distribution on [0,1],
for large $k$, density $f_{k-1}$ will concentrate around the value of $c$ for which
$M_{k-1}(c)$ is largest.

\begin{remark}\label{filtration}
In practice the observation of an rv $T_k$ is accompanied by some, possibly random, attributes
like the real time and the monetary cost needed to determine the value of $T_k$.
In particular the actual filtration that one would like to adopt is much richer than
$\{\mathcal F_k^o=\sigma(T_1,\ldots, T_k)\}_k$, 
and should include available real time and monetary budget,
and possibly the mental condition of the investigator.
In order to stay close to the intuition for a random, iid, sample $T_1,T_2,\ldots$, 
a suitable extra condition
on the sample is that $T_k,T_{k+1},\ldots$ and their attributes are independent 
of all information contained in $\mathcal F_{k-1}$.
One can reach this by actually \emph{hiding} all information about the rv's $T_\ell$ for $\ell\ge k$ and their attributes
until the decision to determine and process the value of $T_k$. % including the processing the observed value of $T_k$. 
On the other hand, enriching of the filtration typically allows for an
$\mathcal F_{k-1}$ measurable rv $c_{k-1}$, depending not only on $T_1,\ldots,T_{k-1}$,
but for example also on the built-up insights of the investigator up to time $k-1$. 
\end{remark}
We have presented the theory in its purely sequential form.
In practice it may occur that sampling takes place in the form of 
a sequence of random samples $S_1=\{T_1,\ldots,T_{n_1}\}, S_2=\{T_{n_1+1},\ldots,T_{n_1+n_2}\},\ldots$
of rv $T$.
Associated to this is a filtration $\mathcal F_0,\mathcal F_1,\mathcal F_2,\ldots$ where
$n_1$ is $\mathcal F_0$-measurable, $T_1,\ldots,T_{n_1}$ and $n_2$ are $\mathcal F_1$-measurable and so on.
Moreover one would like $S_1$ independent of $\mathcal F_0$, $S_2$ independent of $\mathcal F_1$, etc.
A test martingale may be constructed by choosing $c_0$ $\mathcal F_0$-measurable, $c_1$ $\mathcal F_1$-measurable
and having $M_0=1$,
$$
M_1=M_0\times\prod_{i=1}^{n_1}(1-c_0\frac{T_i-\mu}{\tau-\mu})
,\qquad
M_2=M_1\times\prod_{i=n_1+1}^{n_1+n_2}(1-c_1\frac{T_i-\mu}{\tau-\mu}),\qquad \ldots
$$
Thus we have constructed a test martingale for $H_0:\E(T)\ge\mu$.
Of course there is also an integrated version: $M_0(\pi)=1$,
$$
M_1(\pi)=\int\prod_{i=1}^{n_1}(1-c\,\frac{T_i-\mu}{\tau-\mu})\pi(c)dc
,\quad
M_2(\pi)=\int\prod_{i=1}^{n_1+n_2}(1-c\,\frac{T_i-\mu}{\tau-\mu})\pi(c)dc,\quad \ldots
$$
This procedure is less efficient than the sequential procedure, mainly because
the above $M_1, M_2, \ldots$ occur 
as $M_{n_1},M_{n_1+n_2},\ldots$ in a test martingale corresponding to 
a sequential procedure based on $T_1,T_2,T_3,\ldots$,
so that intermediate opportunities to reject $H_0$ are missed.

\subsubsection*{Test martingales for the null-hypothesis $H_0:\E(T)\le\mu$ or $H_0:\E(T)=\mu$}
Given $\tau_0\le T$ a.s. and $\mu>\tau_0$, we obtain test martingales for the
null-hypothesis $H_0:\E(T)\le\mu$ by transforming it into $H_0:\E(-T)\ge-\mu$
leading to multiplication factors
\begin{equation}\label{lower bound}
1-c_{k-1}\,\frac{-T_k+\mu}{-\tau_0+\mu},~k=1,2,\ldots,
\end{equation}
with $0\le c_{k-1}\le1$.
As a curiosity we remark the following 
\begin{remark}
Suppose $\E(T)\le\nu<\mu$ and $\tau_0\le T\le\tau_1$ a.s.. 
Let $\{M_k(c)\}_k$ be the test martingale for $H_0:\E(T)\ge\mu$ as constructed before.
If $c\le(\mu-\nu)/(\mu-\tau_0)$, then $\{(M_k(c))^{-1}\}_{k=0}^\infty$ is a supermartingale.
In particular it is a test martingale for null hypothesis $H_0:\E(T)\le(1-c)\mu+c\tau_0$.
%In particular, if one observes $\exists k:M_k(c)\le\alpha$ one might reject
%at significance level $\alpha$ the hypothesis
%that  $\{(M_k(c))^{-1}\}_{k=0}^\infty$ is a supermartingale, or that $\E(T)\le\mu-c(\mu-\tau_0)$.
\end{remark}
The test martingale $\{(M_k(c))^{-1}\}_{k=0}^\infty$ for $H_0:\E(T)\le(1-c)\mu+c\tau_0$, 
as in the above Remark,
is not optimal for this $H_0$ since the factors $(1-c\,(T_k-\mu)/(\tau_1-\mu))^{-1}$
are concave in $T_k$ rather than linear.
Nevertheless, if $M_k(c)^{-1}\ge\alpha^{-1}$, i.e. $M_k(c)\le\alpha$, 
it signals that $\E(T)>(1-c)\mu+c\tau_0$. % is close to or larger than $\mu$.
\\
\\
\begin{proof}[Proof of remark]
It follows from the analog of Lemma \ref{Jensen+} for the convex function 
$t\mapsto f_c(t)$ with $f_c(t)=(1-c(t-\mu)/(\tau_1-\mu))^{-1}$.
Namely, suppose $\E(T)=\nu^*\le\nu$, the maximal value of $\E(f_c(T))$ is attained at the alternative distributed rv $T^A$
with values $\tau_0$ and $\tau_1$.
One checks easily that $\E(f_c(T^A))\le1$ if $c=(\mu-\nu^*)/(\mu-\tau_0)$ and if $c=0$.
Since $\E(f_c(T^A))$ is a convex function of $c$ we have $\E(f_c(T^A))\le1$
for $0\le c\le (\mu-\nu^*)/(\mu-\tau_0)$.
It follows that $\E(f_c(T))\le1$ if $0\le c\le (\mu-\nu)/(\mu-\tau_0)\le (\mu-\nu^*)/(\mu-\tau_0)$, which condition entails $\nu^*\le(1-c)\mu+c\tau_0$.
\end{proof}
\\
\\
One can combine a test martingale $\{M_n^+\}_n$ for $H_0:\E(T)\ge\mu$ and 
a test martingale $\{M_n^-\}_n$ for $H_0:\E(T)\le\mu$, based on the same data,
by taking $\{\rho^+ M_n^++\rho^-M_n^-\}_n$ for any $\rho^+,\rho^-\ge0$ with $\rho^++\rho^-=1$.
This will be a test martingale for $H_0:\E(T)=\mu$.

%
%
%%%%%%%%%%%%%%%%%%%%%%%%%%%%%%%%%%%%%%%%%%%%%%%%%%%%%%%%%%%%%%%%%%%%%%%%%%%%%%
\section{Behavior of the test martingales $\{M_k(c)\}_k$}\label{Behavior}
We discuss the behavior of the test martingales $\{M_k(c)\}_k$ constructed with $c_{k-1}=c$,
for all $k$, and given $0\le c<1$. 
Let $T$ be an integrable random variable such that $T\le\tau$ a.s..
We have
$$
\frac1k\log(M_k)=\frac1k\sum_{i=1}^k\log(1-c\,(T_i-\mu)/(\tau-\mu)).
$$
Consider the function 
$$
\lambda(c)=\E(\log(1-c\,(T-\mu)/(\tau-\mu)),\quad 0\le c<1.
$$
Because of the concavity of the $\log$ function we have
$$\log(1-c)\le \E(\log(1-c\,(T-\mu)/(\tau-\mu))\le\log(1-c\,(\E(T)-\mu)/(\tau-\mu)).$$
%For $0\le c<1$ let $\lambda(c)=\E(\log(1-c\cdot(T-\mu)/(\tau-\mu))$.
Let $Z=(T-\mu)/(\tau-\mu)$ then $\E(Z)$ exists and 
$Z\le1$ a.s..
Supposing differentiation with respect to $c$ behaves decently
with respect to expected value, we have
\begin{align*}
\lim_{c\downarrow0}\lambda(c)&=\lambda(0)=0;
\\
\lambda'(c)
&=
\E(-Z/(1-c\, Z));
\\
\lim_{c\downarrow0}\lambda'(c)
&=
\E(-Z)=-(\E(T)-\mu)/(\tau-\mu);
\\
\lambda''(c)
&=
\E(-[Z/(1-c\, Z)]^2)\le0.
\end{align*}
%In order to enforce continuity of $\lambda'(c)$ at $c=0$ we will
%require that $T$ is bounded from below, $\tau_0\le T$ a.s. for some $\tau_0$.
Notice that the first and second derivatives in $c$ for $0<c<1$ of 
$\log(1-c\, z)$
are uniformly bounded in $z\in(-\infty,1]$. %, $-(1-c)^{-1}\le\lambda'(c)\le1/c$.
In particular the above expected values do exist for $0<c<1$
and 
$\lambda(c)$ is a concave function in $c\in(0,1)$.
It is clear that
$$-\frac1{1-c}\le\frac{-Z}{1-c\,Z}\le-Z$$
so that by dominated convergence it follows that $\lim_{c\downarrow0}\E(-Z/(1-cZ))=\E(-Z)$.

\begin{lemma}\label{L1}
% For~ $0\le c<1$ define $\lambda(c)=\E(\log(1-c\cdot(T-\mu)/(\tau-\mu)))$.
%Then 
For $T$ as above,
$\lambda$ is a differentiable concave function on $(0,1)$,
continuous on $[0,1)$
and $\lim_{c\downarrow0}\lambda'(c)=-(\E(T)-\mu)/(\tau-\mu)$.
If $\P\{T=\mu\}<1$ it is strictly concave.
%Consider the test martingale $\{M_n(c)\}_n$ constructed with the multiplication factors
%$(1-c\cdot(T-\mu)/(\tau-\mu))$
%and let $L_n(c)=\log(M_n(c))$.
Consider $%L_n:c\mapsto 
	L_n(c)=\log(M_n(c))$.
Any realization 
%$\ell_n$ 
of $L_n(c)$ (based on observations of $T_1,\ldots,T_n$) 
is concave in $c$.
Thus $M_n(c)$ is a log-concave in $c$.
\end{lemma}

\begin{corollary}\label{C2}
Suppose $\P\{T=\mu\}<1$.
If $\E(T)\ge\mu$ and $0<c<1$, then $\lambda(c)<0$, so that $\lim_{n\to\infty}M_n(c)=0$.
If $\E(T)<\mu$, then there is $c_{\max}>0$ such that for $0<c<c_{\max}$ we have $\lambda(c)>0$
implying that $\lim_{n\to\infty}M_n(c)=\infty$ a.s..
\end{corollary}

\begin{proof}
Suppose $\E(T)=\mu$, but $\P\{T=\mu\}<1$, then $\lambda'(0)=0$ and 
from the strict concavity of $\lambda(c)$ in $c$, it follows that $\lambda(c)<0$ for all $0<c<1$.
If $\E(T)>\mu$, then $\lambda'(0)<0$ and again  it follows that $\lambda(c)<0$ for all $0<c<1$.
According to the strong law of large numbers, it follows that
$\lim_{k\to\infty}\frac1k\log(M_k(c))=\lambda(c)<0$ a.s., 
and therefore that $\lim_{k\to\infty}M_k(c)=0$ a.s. 
(despite the fact that $\E(M_k)=1$ for all $k$, cf. \cite[Ex. 5.2.9]{D}).
If $\E(T)<\mu$, then $\lambda'(0)>0$ and there is $c_{\max}>0$ such that $\lambda(c)>0$
for $0<c<c_{\max}$.
Now $\lim_{k\to\infty}\frac1k\log(M_k(c))=\lambda(c)>0$ a.s. and therefore
$\lim_{k\to\infty}M_k(c)=\infty$ a.s. for $0<c<c_{\max}$.
\end{proof}
\\
\\
%We try to find sufficient conditions enforcing $\lambda(c)>0$.
Suppose that $\tau_0\le T\le\tau_1$ a.s. and $\tau_0<\E(T)=\nu<\mu<\tau_1$.
Concavity of the function
$t\mapsto\log(1-c(t-\mu)/(\tau_1-\mu))$ implies 
that the minimum with respect to the $T$-distribution of $\lambda(c)=\E(\log(1-c(T-\mu)/(\tau_1-\mu))$
is obtained for the distribution $T^A$ concentrated at the endpoints $\tau_0$ and $\tau_1$
of the support such that $\E(T^A)=\nu$ (see Lemma \ref{Jensen+}).
Then $\P(T^A=\tau_1)=(\nu-\tau_0)/(\tau_1-\tau_0)$, $\P(T^A=\tau_0)=(\tau_1-\nu)/(\tau_1-\tau_0)$
and
\begin{align*}
%\lambda(c)=
\lambda^A(c)
&=\E(\log(1-c(T^A-\mu)/(\tau_1-\mu))
\\&=
\frac{\nu-\tau_0}{\tau_1-\tau_0}\log(1-\frac c{\tau_1-\mu}(\tau_1-\mu))
+
\frac{\tau_1-\nu}{\tau_1-\tau_0}\log(1-\frac c{\tau_1-\mu}(\tau_0-\mu)).
\end{align*}
The maximum of $c\mapsto\lambda^A(c)$ is attained at the unique point $c$ where the derivative with respect to $c$ is zero,
that is $c=(\mu-\nu)/(\mu-\tau_0)$ and then
% the multiplication factors in the test martingale are
$$1-c(\tau_1-\mu)/(\tau_1-\mu))=((\nu-\tau_0)/(\tau_1-\tau_0))/((\mu-\tau_0)/(\tau_1-\tau_0))$$
$$1-c(\tau_0-\mu)/(\tau_1-\mu))=((\tau_1-\nu)/(\tau_1-\tau_0))/((\tau_1-\mu)/(\tau_1-\tau_0))$$
so that then $\lambda^A((\mu-\nu)/(\mu-\tau_0))$ equals the 
Kullback-Leibler divergence 
$$D_{\hbox{KL}}(\Alt((\nu-\tau_0)/(\tau_1-\tau_0))\,\|\,\Alt((\mu-\tau_0)/(\tau_1-\tau_0)))$$
from $\Alt((\mu-\tau_0)/(\tau_1-\tau_0))$ to $\Alt((\nu-\tau_0)/(\tau_1-\tau_0))$.

\begin{lemma}\label{L2}
Suppose $\tau_0\le T\le \tau_1$ a.s., and $\tau_0<\E(T)=\nu<\mu<\tau_1$.
Then 
$$\lambda((\mu-\nu)/(\mu-\tau_0)))\ge D_{\hbox{KL}}(\Alt((\nu-\tau_0)/(\tau_1-\tau_0))\,\|\,\Alt((\mu-\tau_0)/(\tau_1-\tau_0)))>0,$$
and $\lambda'((\mu-\nu)/(\mu-\tau_0)))\ge0$.
In particular $\lambda$ reaches its maximum at some $c_{\rm opt}$ 
with $(\mu-\nu)/(\mu-\tau_0)\le c_{\rm opt}\le1$.
If $\P\{T=\tau_1\}>0$, then $\lim_{c\uparrow1}\lambda(c)=-\infty$.
\end{lemma}

\begin{proof}
We have already shown the lower bound for $\lambda((\mu-\nu)/(\mu-\tau_0)))$.
The last claim of the Lemma follows from the upper bound 
$$\lambda(c)\le\P\{T=\tau_1\}\log(1-c)+\log(1-(\tau_0-\mu)/(\mu-\tau_1).$$ 
To prove that $\lambda'((\mu-\nu)/(\tau_1-\mu))\ge0$,
consider $Z=(T-\mu)/(\tau_1-\mu)$. We have seen that
$$
\lambda'(c)=\frac{\partial}{\partial c}\E(\log(1-cZ))=\E(-Z/(1-cZ)).%\ge0.
$$
Now $z\mapsto -z/(1-cz)$ satisfies 
$\frac{\partial^2}{\partial z^2} (-z/(1-cz))=-2c/(1-cz)^3<0$
so that it is a concave function.
In particular the minimal value of $\lambda'(c)=\E(-Z/(1-cZ))$ with respect to distributions of $T$
with $\tau_0\le T\le \tau_1$ a.s. and $\E(T)=\nu$, is attained at rv $T^A$ having the alternative distribution 
with support in $\{\tau_0,\tau_1\}$,
in which case
we have seen that $\lambda^A$ attains its maximum at $c=(\mu-\nu)/(\mu-\tau_0))$.
So 
$\lambda'((\mu-\nu)/(\mu-\tau_0))\ge(\lambda^A)'((\mu-\nu)/(\mu-\tau_0))=0$. 
% so that in the general situation $\lambda'((\mu-\nu)/(\tau_1-\mu))\ge0$.
From lemma \ref{L1} we know that $\lambda(c)$ is a concave function in $c$.
In particular, if $\lambda'(c)\ge0$, then the maximum is attained at some $c_{\rm opt}\ge c$.
\end{proof}
\\
\\
%We conclude as follows:
Suppose that $0 \le T\le 1$ a.s. and $\E(T)=\nu<\mu$.
For which $c$ do we have
$$\nu\log(1-c\cdot(1-\mu)/(1-\mu))+(1-\nu)\log(1-c\cdot(0-\mu)/(1-\mu))>0  ~ ? $$
%If we fix such $c$, the power of the test procedure is 1.
\iffalse
mu<-0.05;nu<-0.02
b<-seq(0,0.95,length=100)
y<-nu*log(1-b*(1-mu))+(1-nu)*log(1-b*(-mu))
plot(b,y,type="l")
uniroot(function(b)nu*log(1-b*(1-mu))+(1-nu)*log(1-b*(-mu)),c(0.001,1))
optim(0.5,function(b)-(nu*log(1-b*(1-mu))+(1-nu)*log(1-b*(-mu))))

\fi
For $\mu=0.05$ and $\nu=0.02$, and $T=\Alt(\nu)$ we find $0<c<c_{\max}=0.895$ and 
we find a maximum 0.012 at $c=c_{\rm opt}=(\mu-\nu)/\mu=0.60$.
%
%
%
\iffalse
mu<-0.05;d<-mu/(1-mu); # dan mu==d/(1+d)
cs<-seq(0,1-1e-5,length=1000)
plot(cs,1-cs,type='l',ylab='nu/mu')
lines(cs,(1+d)/d*log(1+cs*d)/(log(1+cs*d)-log(1-cs)))
d<-0.1
lines(cs,(1+d)/d*log(1+cs*d)/(log(1+cs*d)-log(1-cs)),col='red')

\fi
%
%
%

%%%%%%%%%%%%%%%%%%%%%%%%%%%%%%%%%%%%%%%%%%%%%%%%%%%%%%%%%%%%%%%%%%%%%%%%%%%%%%
\subsubsection*{Probability of Type I error}
Given a test martingale $\{M_k\}_k$ for testing $H_0:\E(T)\ge\mu$ 
at the signifcance level $\alpha$ we reject $H_0$ at the time $k$ that
$M_k\ge1/\alpha$.
It is clear that the size (i.e. probability of Type I error) of the test is %less than 
$\P\{\exists k:M_k\ge1/\alpha\}\le\alpha$.
We will give a lower bound for the size % probability of rejecting $H_0:\E(T)\ge\mu$,
under the additional conditions that the support of $T$ is bounded from below, that $\E(T)=\mu$
and $\P(T=\mu)<1$.

\begin{theorem}\label{size}
Let $T$ be a random variable such that $\tau_0\le T\le\tau_1$ a.s., and let $0<c<1$
and consider the test martingale $\{M_k(c)\}_k$ with 
multiplication factor $(1-c(T_k-\mu)/(\tau_1-\mu))$
at time $k$.
Suppose $\E(T)=\mu$, but $\P\{T=\mu\}<1$.
%Then $\lim_{t\to\infty}M_t=0$ a.s.
%Let $M_t^\alpha=M_t$ if $M_s<1/\alpha$ for all $s\le t$, and $M_t^\alpha=M_r$
%if $r=\min\{s;M_s>1/\alpha\}$.
%Then $\{M_t^\alpha\}$ is a martingale, and its almost sure limit $M^\alpha_\infty$ exists. 
\\
The probability of Type I error of the test is less than $\alpha$ but greater than 
$\alpha/(1-c(\tau_0-\mu)/(\tau_1-\mu))>\alpha(\tau_1-\mu)/(\tau_1-\tau_0)$.
\end{theorem}
In particular, if $\E(T)=\mu$ and $(\mu-\tau_0)/(\tau_1-\tau_0)$ is small,
the null hypothesis $H_0:\E(T)\ge\mu$ wil be rejected with
probability close to  (but less than) $\alpha$.

\begin{proof}
The process $\{M_k\}_k=\{M_k(c)\}_k$ is a martingale.
Let $M^\alpha_k=M_k$ if $M^*_k<1/\alpha$ and $M^\alpha_k=M_n$ if $k\ge n$,
$M_n\ge1/\alpha$
and $M_{n-1}^*<1/\alpha$.
Then $\{M^\alpha_k\}_k$ is a stopped martingale, and therefore a martingale.
Let $M^\alpha_\infty=\lim_{k\to\infty}M^\alpha_k$, 
then (see Cor. \ref{C2})  $M^\alpha_\infty$ takes values in the 
set $\{0\}\cup[\alpha^{-1},(1-c(\tau_0-\mu)/(\tau_1-\mu))\alpha^{-1})$. 
Because of dominated convergence we have
$1=\E(M_0)=\lim_{k\to\infty}\E(M^\alpha_k)=\E(M^\alpha_\infty)$
and 
$\alpha^{-1}\P\{M^\alpha_\infty\ge\alpha^{-1}\}\le\E(M^\alpha_\infty)<
(1-c(\tau_0-\mu)/(\tau_1-\mu))\alpha^{-1}\P\{M^\alpha_\infty\ge\alpha^{-1}\}$.
\end{proof}

%%%%%%%%%%%%%%%%%%%%%%%%%%%%%%%%%%%%%%%%%%%%%%%%%%%%%%%%%%%%%%%%%%%%%%%%%%%%%%
\subsubsection*{Power}
First we present a result showing that power equal to 1 is attainable.

\begin{theorem}[Consistency]\label{Consistency}
Suppose $\E(T)<\mu<\tau_1$ and $T\le\tau_1$ a.s..
Then there is $c_{\max}>0$ such that $\lim_{n\to\infty} M_n(c)=\infty$ a.s. for $0<c<c_{\max}$.
Let $\pi$ be a probability density on $[0,1]$ such that $\pi(c)>0$ for $c\in(0,\varepsilon)$
for some $\varepsilon>0$.
%$\int_0^\varepsilon d\pi(c)>0$ for all $\varepsilon>0$.
Then the integrated test martingale $\{M_n(\pi)\}_n$
satisfies $\lim_{n\to\infty} M_n(\pi)=\infty$ a.s..
In particular, the test based on $\{M_n(\pi)\}_n$ is consistent, i.e. the power of the test is 1.
\end{theorem}
\begin{proof}
The first claim follows from Corollary \ref{C2}.
Let $0<a<b<\min(\varepsilon,c_{\max})$, then $p=\int_a^b\pi(c)dc>0$.
We have $\lim_{n\to\infty} M_n(c)=\infty$ a.s. for $c=a,b$.
Be given any $R>0$,
let $N$ be such that $M_N(a)>R/p$ and $M_N(b)>R/p$. 
Since $M_N(c)$
is a log-concave function in $c$ we have $M_N(c)>R/p$ for all $c\in[a,b]$, so that
$M_N(\pi)>(R/p)\cdot p=R$.
\end{proof}
\\
\\
The above theorem is applicable if one considers the uniform probability density $\pi$
on [0,1].
If one is convinced that $\E(T)\le\nu<\mu$, one can considerably improve the efficiency
of the test.
According to Lemma \ref{L2} there is some $c=c_{\rm{opt}}$ with 
$(\mu-\nu)/(\mu-\tau_0)\le c_{\rm{opt}}\le1$
for which $\E(\log(1-c(T-\mu)/(\tau_1-\mu)))$ and therefore $\E(\log(M_n(c)))$ is optimal.
This suggests to consider the integrated test martingale with respect to the uniform probability density $\pi$ on $[(\mu-\nu)/(\mu-\tau_0),1]$.
See Table \ref{T1} for its performance.

\begin{theorem}\label{Th1}
Suppose $\tau_0\le T\le \tau_1$ a.s., and $\E(T)\le\nu$ with $\tau_0<\nu<\mu<\tau_1$.
The power of the test with test martingale 
$\{M_k(c)\}_k$ is 1 if $0<c\le(\mu-\nu)/(\mu-\tau_0)$.
Let $\pi$ be a probability density on $[(\mu-\nu)/(\mu-\tau_0),1]$ such that
$\pi(c)>0$ for $c\in[(\mu-\nu)/(\mu-\tau_0),(\mu-\nu)/(\mu-\tau_0)+\varepsilon)$ for some 
$\varepsilon>0$,
% $\int_{(\mu-\nu)/(\mu-\tau_0)}^z\pi(c)dc>0$ for all $z>(\mu-\nu)/(\mu-\tau_0)$,
%for all $c$, 
then the test based on $\{M_n(\pi)\}_n$ has power 1.
\end{theorem}
%
%

\iffalse
mu<-0.05;M<-100
p<-M/(M+1)
eps<-2*mu*(1-mu)*M/(1+M)/(M+mu^2)
p*log(1-eps)+(1-p)*log(1-eps/(1-mu)*(-M-mu))
uniroot(function(eps)p*log(1-eps)+(1-p)*log(1-eps/(1-mu)*(-M-mu)),c(mu/M,1),tol=1e-10)
eps<-seq(0,2*mu/(1-mu)/(mu+M),length=1000);plot(eps,p*log(1-eps)+(1-p)*log(1-eps/(1-mu)*(-M-mu)),type='l')

\fi
%%%%%%%%%%%%%%%%%%%%%%%%%%%%%%%%%%%%%%%%%%%%%%%%%%%%%%%%%%%%%%%%%%%%%%%%%%%%%%
\subsubsection*{Average sample number}
We took $\mu=0.05$, significance level $\alpha=0.05$ 
and considered the necessary sample number for rejection of $H_0$ for different 
$T$-distributions with $0=\tau_0\le T\le\tau_1=1$, $\E(T)=0.02$ and 
test martingales $\{M_k(c)\}_k$ with $0\le c\le 1$.
The results are compiled in Table \ref{T1}.
The last line starting with $[0.6,1]$ is based on the integrated test martingale 
$\{M_k(\pi)\}_k$ 
over the uniform density $\pi$ on the interval $[0.6,1]$, 
possibly based on a strong conviction
that $\E(T)\le\nu=0.02$, implying that the optimal $c$ is not less than $(\mu-\nu)/\mu=0.6$.
The mentioned average sample numbers thus correspond to tests with power 1.
The average sample number and standard deviations are in each instance based on 1000 test runs.
\\
\noindent
\begin{table}[h]
\begin{tabular}{lcc}
$T:$&\multicolumn{2}{c}{$\Alt(0.02)$}\\
\phantom{0}$c$&\multicolumn{1}{c}{mean}&sd\\
0.2 & 516.2 & 127.4 \\
0.4 & 294.1 & 124.4 \\
0.6 & 245.9 & 169.2 \\
0.8 & 357.7 & 510.5 \\
1 & $\infty$ & -- \\
%\multicolumn{2}{l}{uniform on}
%\\
$[0.6,1]$ & 287.6&253.2
\end{tabular}
\begin{tabular}{cc}
\multicolumn{2}{c}{$5T\sim\Alt(0.10)$}\\
\multicolumn{1}{c}{mean}&sd\phantom{0}\\
482.7 & 45.0 \\
245.5 & 32.7 \\
166.7 & 27.7 \\
127.4 & 24.9 \\
104.0 & 23.4 \\
124.6 & 25.4
\end{tabular}
%\qquad
\begin{tabular}{rr}
\multicolumn{2}{c}{$\hbox{Beta}(0.02,0.98)$}\\
\multicolumn{1}{c}{mean}&sd\phantom{0}\\
495.3 & 81.7 \\
261.3 & 67.8 \\
186.9 & 68.0 \\
156.1 & 79.8 \\
166.2 & 155.0 \\
162.4&85.5
\end{tabular}
\begin{tabular}{rr}
\multicolumn{2}{c}{$\hbox{Beta}(2,98)$}\\
\multicolumn{1}{c}{mean}&sd\phantom{0}\\
476.6 & 10.1 \\
239.6 & 7.1 \\
160.5 & 5.9 \\
121.0 & 5.1 \\
97.2 & 4.5 \\
117.7&5.3
\end{tabular}
\begin{tabular}{r}
$0.02$\\
\multicolumn{1}{c}{=}\\
476 \\
239 \\
160 \\
121 \\
97 \\
117
\end{tabular}
\caption{Average sample numbers (mean) and their standard deviations (sd), 
in each case based on 1000 tests of $H_0:\E(T)\le0.05$ with significance level $\alpha=0.05$.}
\label{T1}
\end{table}
\\
We see confirmed that the optimal $c$ depending on the distribution of $T$ is some number greater than
$(\mu-\nu)/\mu=0.6$ (see Lemma \ref{L2}).
%Our conclusion is that if tainting $\tau_1=1$ does regularly occur, one should choose $c$ close to 
Notice that it is required that $c\le1$ 
in order to have a test martingale.
Notice that, except for the alternative distribution, the integrated test martingale
over the interval $[(\mu-\nu)/\mu,1]$ outperforms
the test martingale corresponding to a fixed parameter $c=(\mu-\nu)/\mu$.
If one has no idea about $\E(T)$ other than $\E(T)<\mu$, then the integrated test martingale, 
integrated uniformly over $[0,1]$,
is a suitable choice.
\\
One may compare the results in Table \ref{T1} with a common practice in financial auditing
as expressed in the Audit Guide Audit Sampling (AICPA, 2012).
With $\alpha$, $\mu$ and $\nu$ as above, Table 4-5 (or C-1) lead to
a sample size of 162. 
One will reject $H_0$ if
the total sum of $T$-values does not exceed 3.24 
(which is the expected value if $\E(T)=\nu=0.02$).
\\
\\
We make a side step to Wald's equation (\cite[Thm. 4.1.5]{D}).
Consider the test martingale $\{M_k(c)\}_k$ constructed as above
with constant $c$ and let $Z=(T-\mu)/(\tau_1-\mu)$. 
Suppose $\E(\log(1-cZ))>0$, so that the power of the test is 1.
Consider the stopping time $N$ where $\{M_k\}_k$ crosses level $1/\alpha$ for the first time.
Since $\E(\log(M_k))=k\E(\log(M_1))=k\,\E(\log(1-cZ))$ and
$\{\log(M_k)-\E(\log(M_k))\}_k$ is a martingale, at stopping time $N$
we have 
$\E(\log(M_N)-\E(\log(M_N)))=\E(\log(M_0)-\E(\log(M_0)))=0$.
So $\log(1/\alpha)-\E(N)\E(\log(1-cZ))\ge0$, that is $\E(N)\le\log(1/\alpha)/\E(\log(1-cZ))$.
\\
It is known as Wald's second equation
(\cite[Thm. 4.1.6]{D}) that
$$\E\left[(\log(M_N)-\E(\log(1-cZ))N)^2\right]=\Var(\log(1-cZ))\E(N)$$
from which 
$$\Var(N)\approx\frac{\Var(\log(1-cZ))\E(N)}{\E(\log(1-cZ))^2}
\approx\frac{\Var(\log(1-cZ))}{\E(\log(1-cZ))^3}\log(1/\alpha).$$
These approximations are in reasonable agreement with the above table.
One reason for this is that 
$0\le M_N-\alpha^{-1}<\alpha^{-1}(1-c(T_N-\mu)/(1-\mu)-1)\le\alpha^{-1}\mu/(1-\mu)\approx0.053\alpha^{-1}$
so that $\alpha^{-1}$ is a good approximation of $M_N$ (cf. Theorem \ref{size}).

%%%%%%%%%%%%%%%%%%%%%%%%%%%%%%%%%%%%%%%%%%%%%%%%%%%%%%%%%%%%%%%%%%%%%%%%%%%%%%
\section{Confidence regions}\label{conf reg}

As one may have noticed we did not include a provision in our tests to avoid infinite sample size,
especially in case $H_0$ is satisfied.
In practice it may be a more important issue to find a suitable confidence upper bound and/or confidence lower bound.

\subsubsection*{Confidence bounds}
Choose a confidence level $(1-\alpha)$ with $0<\alpha<1$, for example $\alpha=0.05$.
We will construct an adapted process $\{U^r_k\}_k$ of $(1-\alpha)$-confidence upper bounds
such that even $\P\{\exists k:\E(T)\ge U_k^r\}\le\alpha$.
Thus one can evaluate sequentially the $U^r_k$ and stop whenever one likes and propose the
minimal value found.
\\
\\
Suppose for each $\mu<\tau_1$ we have maintained a process $\{M^\mu_k\}_k$, such that
$\{M^\mu_k\}_k$ is test martingale for the hypothesis $H_0:\E(T)\ge \mu$ given that $T\le\tau_1$ a.s..
Suppose moreover that $\mu\mapsto M^\mu_k$ is continuous and increasing 
for all $k$ and all values of $T_1,\ldots,T_k$. % (possibly excluding a zero probability subset).
We will call such a family $\{M^\mu_k\}_{k,\mu}$ of test martingales a \emph{suitable} family.
Let $U^r_k$ be the statistic depending on $T_1,\ldots,T_k$ taking the value $\mu^r$ where
$$\mu^r=\infty \hbox{ if } M^\mu_k<1/\alpha\hbox{ for all }\mu,\hbox{ and }
\mu^r=\inf\{\mu\mid M^\mu_k\ge1/\alpha\}\hbox{ otherwise. }
$$
If $U^r_k=\mu^r<\infty$, then it follows from continuity that $M_k^{\mu^r}\ge1/\alpha$.
Then for $p\le\tau_1$ we have identity of events 
$$\{\exists k:p\ge U^r_k\}=\{\exists{k}:M^p_k\ge1/\alpha\}.$$%\subseteq\{\sup_{k}M^p_k\ge1/\alpha\}.$$ 
Suppose $\E(T)=p$,
then it follows that $\P\{\exists k:p\ge U^r_k\}=\P\{\exists{k}:M^p_k\ge1/\alpha\}\le\alpha$.

\begin{theorem}
%Let $T$ be as above, then
The statistics $U^r_k$ satisfy
$\P\{\forall k:\E(T)<U^r_k\}\ge1-\alpha$.
\end{theorem}
Thus, at time $k$, $\min\{U^r_\ell\mid\ell\le k\}$ is a $(1-\alpha)$-confidence upper bound.
Moreover, if its value at time $k$ is not convenient, one may continue sampling
in the hope to find
a lower upper bound.
\\
In order to construct a suitable family of test martingales,
consider the test martingale multiplication factor
$$\rho=1-c\,(T-\mu)/(\tau_1-\mu),$$
where $c$ may depend on $\mu$.
%If $b\ge0$ is not depending on $\mu$, then clearly $\rho$ is increasing with respect to $\mu$.
We want it increasing in $\mu$, for all values of $T$. 
This is the case if $\tau_0\le T\le\tau_1$ a.s. and $c=c(\mu)$ depends differentiably on $\mu$ and satisfies:
$$
0\le c(\tau_1)\le c(\tau_0)\le1
\hbox{ and }
0\le-\frac{\partial}{\partial\mu}\log(c(\mu))\le\frac1{\tau_1-\mu}+\frac1{\mu-\tau_0}.
$$
We present two examples.
\begin{itemize}
\item
One may take $c$ independent of $\mu$, resulting in the test martingales $\{M^\mu_k(c)\}_k$. 
%\item
Instead of a fixed value of $c$,
one may choose a fixed probability density $\pi$ on $[0,1]$, and take the integrated test martingales 
$\{M^\mu_k(\pi)\}_k$ with $M^\mu_k(\pi)=\int M^\mu_k(c)\pi(c)dc$.
\item
If $\tau_0\le T$ a.s. and one considers only upper bounds $\mu\ge m+\tau_0$,
one may take $c=d\,(\tau_1-\mu)^{r}/(\mu-\tau_0)^{s}$ for some $0\le r\le1$
and $0\le s\le1$
and $0\le d\le m^s/(\tau_1-\tau_0-m)^{r}$.
Instead of a fixed value of $d$,
one may consider
integrated test martingales, integrated over $d$ with respect to some probability density on 
$[0,m^s/(\tau_1-\tau_0-m)^{r}]$.
\end{itemize}
%
%
%
%
%
%
\iffalse
# Confidence upper bounds
# feed Alt(0.5)
N<-100
cc<-0.9;nu<-0.07;alpha<-0.05
x<-runif(N)<nu
Msup<-function(mu)max(cumsum(log(1-cc/(1-mu)*(x-mu))))
uniroot(function(mu)Msup(mu)-log(1/alpha),c(0,0.99),tol=1e-8)
reps<-10000;res<-numeric(reps)
for(i in 1:reps){
	x<-runif(N)<nu;res[i]<-
	(uniroot(function(mu)Msup(mu)-log(1/alpha),c(0,0.99),tol=1e-8)$root)}
sum(res>nu)/reps
hist(res)
#################
# Nu met b=cc/mu/(1-mu)
N<-100
cc<-0.015;nu<-0.07;alpha<-0.05
x<-runif(N)<nu
Msup<-function(mu)max(cumsum(log(1-cc/mu/(1-mu)*(x-mu))))
uniroot(function(mu)Msup(mu)-log(1/alpha),c(cc+1e-8,0.99),tol=1e-8)
reps<-10000;res<-numeric(reps)
for(i in 1:reps){
	x<-runif(N)<nu;res[i]<-
	uniroot(function(mu)Msup(mu)-log(1/alpha),c(cc+1e-8,0.99),tol=1e-8)$root}
sum(res>nu)/reps
hist(res)

\fi

\noindent
In the context of Remarks \ref{sharpness} and \ref{sharpness2} 
we present the following example with $c=1$:
\begin{remark}
Suppose $T$ has alternative distribution $Alt(\nu)$, and at time $k$ one has observed $s_k=t_1+\cdots+t_k$.
Choose $\tau_1\ge1$, and consider the test martingales 
constructed with the factors $(1-(t-\mu)/(\tau_1-\mu))=(\tau_1-t)/(\tau_1-\mu)$.
Then $M_k^\mu=(\tau_1-1)^{s_k}(\tau_1)^{k-s_k}/(\tau_1-\mu)^k$.
The corresponding $(1-\alpha)$-confidence upper bound is 
$\mu_k=\tau_1-\tau_1(\alpha ((\tau_1-1)/\tau_1)^{s_k})^{1/k}$.
Considering $(M^\mu)^*_\infty$, 
one sees that even $\mu^*=\min_k\mu_k$ is a $(1-\alpha)$-confidence upper bound.
In case $\tau_1=1$, 
the confidence upper bound $\mu_k$ equals $(1-\alpha^{1/k})$
if $t_1=\ldots=t_k=0$ and this bound equals the minimal value $\mu^*$ if $t_{k+1}=1$.
\end{remark}
%
%
%
%
%We are interested in a stopping rule for determining the search for a convenient
%confidencce upper bound in terms of the observed sample mean. 
In the following remark we show the relation between the upper bound
$U^r_k$ and the sample mean $\overline t_k=\frac1n\sum_{i=1}^kt_i$.
Notice that $\log(M^\mu_k(c))=\sum_{i=1}^k\ell(t_i)$ 
with 
$$\ell(t)=\log(1-c\,(t-\mu)/(\tau_1-\mu)).$$
Function $\ell$ is concave so that by Jensen's inequality 
$\log(M^\mu_k(c))\le k\,\ell(\overline t_k)$,
In particular for $\mu=\overline t_k$ we have $\log(M^\mu_k(c))\le0$, $M^\mu_k(c)\le1$ and
for an integrated martingale
$M^\mu_k(\pi)\le1$.

\begin{remark}\label{Rem7}
Let $\{M^\mu_k(c)\}_k$ be the test martingale with $c_{k-1}=c$ for all $k$.
Suppose that for some $c$ it holds that $M^\mu_k(c)>1$, then 
$\overline t_k<\mu$.
For integrated test martingales: $M_k^\mu(\pi)>1$
also implies $\overline t_k<\mu$.
In particular, with these test martingales $\overline t_k< U^r_k$.
\end{remark}
If at some stopping time $N$ we stop sampling, the status of the confidence upper bound
$\min\{U^r_k\mid k\le N\}$ is clear from the considerations above but
it is not clear in what sense $\overline t_N$ would approximate $\E(T)$.
\\
\\
We present a table of average sample numbers and average mean taintings in 
Table \ref{T2}.
The null hypothesis $H_0:\E(T)\ge\mu_0$ and desired precision $m$ are fixed at $\mu_0=m=0.05$, expectation $\E(T)\le\nu$ is (correctly) guessed to hold with $\nu= 0.02$, leading to
the use of the uniform probability measure on $[0.6,1]$ for the integrated
test martingales $M_k^\mu(\pi)$.
For each $T$-distribution we simulate 1000 runs, each leading 
(after at least 50 observations to avoid high values of $\overline t_n$) 
to 
a $(1-\alpha)$-confidence upperbound $\mu_n^*=\min\{\mu_k\mid k\le n\}$
such that $\mu_n^*-\overline t_n\le m$, and
we record the average sample number $n$ and average sample mean $\overline t_n$.
These runs are extended until a $(1-\alpha)$-confidence upperbound 
satisfies $\mu_n-\overline t_n\le m$,
and we record again the average sample number $n$ and average sample mean $t_n$.
Then these runs are extended until rejection of $H_0$ of which the average sample number $n$ with $\mu_n\le m$ is recorded. 
\begin{table}[h]
\begin{tabular}{lcc}
$T:$&\multicolumn{2}{c}{$\Alt(0.02)$}\\
      &\multicolumn{1}{c}{avg}&sd\\
$\mu_n^*$ & 99.748& 35.10 \\
$\overline t_n$ & 0.0167 & 0.0132 \\
$\mu_n$ & 121.07 & 43.91 \\
$\overline t_n$ & 0.0157 & 0.0119 \\
$\mu_n\le m$     & 268.31 & 246.07 \\
%\multicolumn{2}{l}{uniform on}
%\\
%$[0.6,1]$ & 268.31 & 246.07
\end{tabular}
\begin{tabular}{cc}
\multicolumn{2}{c}{$5T\sim\Alt(0.10)$}\\
\multicolumn{1}{c}{avg}&sd\phantom{0}\\
71.279 & 1.34 \\
0.0200 & 0.0070 \\
72.026 & 0.28 \\
0.0199 & 0.0070 \\
124.16 & 24.54 \\
%124.16 & 24.54
\end{tabular}
%\qquad
\begin{tabular}{rr}
\multicolumn{2}{c}{$\hbox{Beta}(0.02,0.98)$}\\
\multicolumn{1}{c}{avg}&sd\phantom{0}\\
77.21 & 11.47 \\
0.0195 & 0.0108 \\
82.66 & 12.68 \\
0.0190 & 0.0100 \\
157.44 & 80.37 \\
%157.44& 80.37
\end{tabular}
\begin{tabular}{rr}
\multicolumn{2}{c}{$\hbox{Beta}(2,98)$}\\
\multicolumn{1}{c}{avg}&sd\phantom{0}\\
69.99 & 0.09 \\
0.0200 & 0.0017 \\
69.99 & 0.08 \\
0.0200 & 0.0017\\
117.79 & 5.07 \\
%117.79&5.07
\end{tabular}
\begin{tabular}{r}
$0.02$\\
\multicolumn{1}{c}{=}\\
70 \\
0.02 \\
70 \\
0.02 \\
117 \\
%117
\end{tabular}
\caption{Average sample numbers, average sample means and their standard deviations, 
in each case based on 1000 runs.}
\label{T2}
\end{table}
%\iffalse
~\\
A way to shift from an initial suitable family of test martingales to another one is the following.
In Remark \ref{filtration} the context is explained in which a test martingale may be made 
dependent on insights acquired during its construction. 
Essentially the insights are not allowed to depend on any information about the random variables $T_\ell$ 
that are not yet processed in the martingale.
It is wrong to change at time $n$ the dependence on $T_1,\ldots,T_n$ of the test martingales
$\{M_k^\mu\}_{k=0}^n$, 
in the sense that the process will lose its measurability properties.
\\
Suppose at time $n$, in view of the observed values of $T_1,\ldots,T_n$,  
one wants to shift to an other \emph{suitable} family of test martingales $\{N^\mu_k\}_{k,\mu}$,
to be based on the random sample $T_{n+1},T_{n+2},\ldots$, such that $N^\mu_0=1$ for all $\mu$.
\iffalse
A typical example is a choice at time $n$ of some probability density $\pi_n$ on [0,1] and 
$N^\mu_\ell=M^\mu_{n+\ell}(\pi_n)/M^\mu_n(\pi_n)$ where 
$M^\mu_{n+k}(\pi_n)$ is based on $T_1,\ldots,T_n,T_{n+1},\ldots,T_{n+k}$.
\fi
One may proceed as follows:
%$$M_{n+\ell}^\mu(t_1,\ldots,t_{n+\ell})=M_n^\mu(t_1,\ldots,t_n) N_\ell^\mu(t_{n+1},\ldots,t_{n+\ell}).$$
$$M_{n+\ell}^\mu=M_n^\mu N_\ell^\mu.$$
This will result in a suitable family of test martingales $\{M_k^\mu\}_{k,\mu}$.

\subsubsection*{Confidence intervals}
It is more or less natural to associate confidence intervals for $\E(T)$ with a family of
tests of the hypotheses $H_0:\E(T)=\mu$, where $\mu\in\R$.
The desirable property  of the family of tests then is that 
if $H_0:\E(T)=\mu$ can not be rejected for $\mu=\mu_1$ and $\mu=\mu_2$, 
it will not be rejected for all $\mu$ between $\mu_1$ and $\mu_2$.
We will try to find a family of test martingales 
designed to produce two-sided confidence intervals.
Notice the following convexity property of the test martingales $\{M_n^\mu(c)\}_n$:
\begin{lemma}\label{convexity}
Suppose $t_1,\ldots,t_n\in[\tau_0,\tau_1]$, and $0\le c\le1$.
Then
$M_n^{\mu,+}(c)=\prod_{i=1}^n(1-c(t_i-\mu)/(\tau_1-\mu))$ 
and
$M_n^{\mu,-}(c)=\prod_{i=1}^n(1-c(\mu-t_i)/(\mu-\tau_0))$
are convex functions in $\mu$.

\end{lemma}
\begin{proof}
The factors $(1-c(t_i-\mu)/(\tau_1-\mu))=(1-c)+c(\tau_1-t_i)/(\tau_1-\mu)$
are convex and increasing in $\mu$, so that $M_n^{\mu,+}(c)$ is convex and increasing in $\mu$.
The factors $(1-c(\mu-t_i)/(\mu-\tau_0))=(1-c)+c(t_i-\tau_0)/(\mu-\tau_0)$
are convex and decreasing in $\mu$, so that $M_n^{\mu,-}(c)$ is convex and decreasing in $\mu$.
\end{proof}
\\
\\
\iffalse
We will call a function $f:[\tau_0,\tau_1]\to \R$ connected if for all $x$ the subset 
$\{\mu\in[\tau_0,\tau_1]\mid f(\mu)< x\}$ is a connected subset of $\R$ or empty.
Suppose $\tau_0\le T\le\tau_1$ a.s..
Suppose that we have test martingales $\{M_n^\mu\}$ for the hypotheses $H_0:\E(T)=\mu$.
Suppose moreover that for each realization $\mu\to M_n^\mu$ is connected.
Then $\{\mu\in[\tau_0,\tau_1]\mid \forall k\le n: M_k^\mu<1/\alpha\}$ 
is a confidence $(1-\alpha)$-confidence interval.
It is the intersection of the $n$ connected subsets 
$\{\mu\in[\tau_0,\tau_1]\mid M_k^\mu<1/\alpha\}$.
\fi
Now consider a probability density $\pi$ on $[1,-1]$ and let
\begin{align*}
M_n^\mu(\pi)&=\int_{-1}^1 M_n^\mu(c)\pi(c)dc, \hbox{ where}
\\
M_n^\mu(c)&=\begin{cases}
\prod_{i=1}^n(1-c(T_i-\mu)/(\tau_1-\mu)), \hbox{ if }0\le c\le1
\\
\prod_{i=1}^n(1-c(T_i-\mu)/(\mu-\tau_0)), \hbox{ if }-1\le c\le0.
\end{cases}
\end{align*}
%\fi
\begin{theorem}
$\mu\mapsto M_n^\mu(\pi)$ is a convex function.
The region
$\{\mu\mid \forall n:M_n^\mu(\pi)<1/\alpha\}$ is a $(1-\alpha)$-confidence interval.
Thus, for any $k$, $\{\mu\mid \forall n\le k:M_n^\mu(\pi)<1/\alpha\}$
is a $(1-\alpha)$-confidence interval.
\end{theorem}
\begin{proof}[Proof of Theorem]
It follows from Lemma \ref{convexity} that $\mu\mapsto M_n^\mu(\pi)$ is a convex function.
%and therefore connected. 
Thus for each $n$ the set 
$\{\mu\mid M_n^\mu(\pi)<1/\alpha\}$ is an interval, possibly empty, as well as their intersection
$\{\mu\mid \forall n:M_n^\mu(\pi)<1/\alpha\}$.
The probability that $\E(T)=\nu$ lies in the intersection of all $(1-\alpha)$-confidence regions
is the probability that $M_n^\nu(\pi)<1/\alpha$ for all $n$.
Since $M_n^\nu(\pi)$ is a supermartingale, this probability is at least $(1-\alpha)$.
\end{proof}
\\
\\
It is an uncommon and unpleasant feature of sequential procedures 
that it may happen (probability at most $\alpha$)
that the confidence interval is empty.
If one would like to avoid weird decisions, one could stick to one of the
confidence intervals $\{\mu\mid \forall n\le k:M_n^\mu(\pi)<1/\alpha\}$ that is not empty.
In the following we remark that $\overline t_n\in\{\mu\mid M_n^\mu(\pi)<1/\alpha\}$, 
where as usual
$\overline t_n$ denotes the sample average $\frac1n(t_1+\ldots+t_n)$.

\begin{remark}
(Cf. Remark \ref{Rem7})
Let
$\ell(t)=\log(1-c\,(t-\mu)/(\tau_1-\mu)$ if $c\ge0$ and 
$\ell(t)=\log(1-c\,(t-\mu)/(\mu-\tau_0)$ if $c\le0$.
Then $\ell$ is concave. 
From Jensen's inequality follows that $\log(M^\mu_n(c))\le0$ if $\mu=\overline t_n$.
In particular $\overline t_n\in\{\mu\mid M^\mu_n(\pi)<1/\alpha\}$.
\end{remark}

%%%%%%%%%%%%%%%%%%%%%%%%%%%%%%%%%%%%%%%%%%%%%%%%%%%%%%%%%%%%%%%%%%%%%%%%%%%%%%

~\\
\\
Institute for Mathematics, Astrophysics and Particle Physics (IMAPP),
\\
Faculty of Science,
Radboud University Nijmegen,
\\
Heyendaalseweg 135, 6525 AJ Nijmegen, The Netherlands
\\
E-mail: \texttt{H.Hendriks@math.ru.nl}
\end{document}